\documentclass{article}

\usepackage[a4paper]{geometry}

\usepackage{amsmath}
\usepackage{amssymb}
\usepackage{amsfonts}
\usepackage{dsfont}
\usepackage[pagebackref, colorlinks=true, citecolor=blue]{hyperref}
\usepackage{color,graphicx,setspace}
\usepackage{paralist}
\setdefaultitem{\normalfont\bfseries \textendash}{}{}{}
\usepackage{todonotes}
\usepackage{authblk}

\newcommand{\MaxQP}{\textnormal{\textsc{MaxQP}}}
\newcommand{\UnitMaxQP}{\textnormal{\textsc{Unit MaxQP}}}

\newcommand{\MaxCut}{\textsc{MaxCut}}
\newcommand{\val}{\ensuremath{\mathrm{val}}}
\newcommand{\opt}{\ensuremath{\mathrm{opt}}}

\newcommand{\eps}{\varepsilon}
\DeclareMathOperator{\mmod}{mod}

\usepackage{amsthm}
\newtheorem{theorem}{Theorem}
\newtheorem{lemma}{Lemma}

\newtheorem{corollary}{Corollary}

\theoremstyle{definition}

\title{\Large \bf Approximating Sparse Quadratic Programs}

\author[1]{Danny~Hermelin\thanks{The first author's work was supported by the Israel Science Foundation (Grant no.\ 1070/20).}$ ^,$}
\author[2]{Leon~Kellerhals}
\author[2]{Rolf Niedermeier}
\author[1]{Rami~Pugatch\thanks{The fourth author's work was supported by the Israel Science Foundation (Grant no.\ 776/19).}$ ^,$}

\affil[1]{\small Technische Universit\"at Berlin, Chair of Algorithmics and Computational Complexity, Berlin,~Germany, \texttt{leon.kellerhals@tu-berlin.de, rolf.niedermeier@tu-berlin.de}}
\affil[2]{\small Ben-Gurion University of the Negev, Department of Industrial Engineering and Management, Beer~Sheva,~Israel, \texttt{hermelin@bgu.ac.il, rpugatch@bgu.ac.il}}

\date{}

\makeatletter
  \def\abstractname{Abstract.}
  \renewenvironment{abstract}{%
      \if@twocolumn
        \section*{\abstractname}%
      \else
        \small
        \quotation
	\noindent{\bfseries\abstractname}%
      \fi}
      {\if@twocolumn\else\endquotation\fi}
\makeatother

\begin{document}

\maketitle

\begin{abstract}
Given a matrix $A \in \mathbb{R}^{n\times n}$, we consider the problem of maximizing $x^TAx$ subject to the constraint $x \in  \{-1,1\}^n$. This problem, called \MaxQP{} by Charikar and Wirth~[FOCS'04], generalizes \MaxCut{} and has natural applications in data clustering and in the study of disordered magnetic phases of matter. Charikar and Wirth showed that the problem admits an $\Omega(1/\lg n)$ approximation via semidefinite programming, and Alon, Makarychev, Makarychev, and Naor~[STOC'05] showed that the same approach yields an~$\Omega(1)$ approximation when $A$ corresponds to a graph of bounded chromatic number. Both these results rely on solving the semidefinite relaxation of \MaxQP{}, whose currently best running time is $\tilde{O}(n^{1.5}\cdot \min\{N,n^{1.5}\})$, where $N$ is the number of nonzero entries in~$A$ and~$\tilde{O}$ ignores polylogarithmic factors.

In this sequel, we abandon the semidefinite approach and design purely combinatorial approximation algorithms for special cases of \textsc{MaxQP} where~$A$ is sparse (\emph{i.e.}, has $O(n)$ nonzero entries). Our algorithms are superior to the semidefinite approach in terms of running time, yet are still competitive in terms of their approximation guarantees. More specifically, we show that:
\begin{compactitem}
\item \MaxQP{} admits a $(1/2\Delta)$-approximation in $O(n \lg n)$ time, where $\Delta$ is the maximum degree of the corresponding graph.
\item \UnitMaxQP{}, where $A \in \{-1,0,1\}^{n\times n}$, admits a $(1/2d)$-approximation in $O(n)$ time when the corresponding graph is $d$-degenerate, and a $(1/3\delta)$-approximation in $O(n^{1.5})$ time when the corresponding graph has~$\delta n$ edges.
\item \MaxQP{} admits a $(1-\eps)$-approximation in $O(n)$ time when the corresponding graph and each of its minors have bounded local treewidth.
\item \UnitMaxQP{} admits a $(1-\eps)$-approximation in $O(n^2)$ time when the corresponding graph is $H$-minor free.
\end{compactitem}
\end{abstract}

\section{Introduction}

In this paper we are interested in the following (integer) quadratic problem which was coined \MaxQP{} by Charikar and Wirth~\cite{Charikarwirth04}. Given an $n \times n$ symmetric matrix with zero valued diagonal entries $A$, $a_{i,j} \in \mathbb{R}$ for all $i,j \in \{1,\ldots,n\}$, we want to maximize 
\begin{equation}
\label{eqn:maxQP}
\val_x(A)\,\,=\,\, \sum^n_{i=1}\sum^n_{j=1} \, a_{i,j}x_i x_j \quad \text{ s.t. }  
x_i \in \{-1,1\} \text{ for all } i \in \{1,\ldots,n\}.
\end{equation}
Observe that the requirement that all diagonal values of $A$ are zero is to avoid the term $\sum_i a_{i,i}$ which is constant for all choices of~$x_{i,i}$ in (\ref{eqn:maxQP}). Furthermore, a non-symmetric matrix $A$ can be replaced with an equivalent symmetric $A'$ by setting $a'_{i,j}=a'_{j,i}=\frac{1}{2} \cdot (a_{i,j}+a_{j,i})$ without changing~$\val_x(A)$, and so the requirement that~$A$ is symmetric is just for convenience's sake. 

Our interest in \MaxQP{} lies in the fact that it is a generic example of integer quadratic programming which naturally appears in different contexts. Below we review three examples: 

\begin{itemize}
\item \emph{Graph cuts:}
	Readers familiar with the standard quadratic program formulation of \MaxCut{}~\cite{GoemansWilliamson1995} will notice the similarity to (\ref{eqn:maxQP}). Indeed, given a graph $G=(V,E)$ with vertex set $V=\{1,\ldots,n\}$ and edge weights $a_{i,j} \geq 0$ for each $\{i,j\} \in E$, the corresponding \MaxQP{} instance on $-1 \cdot A$ has an optimum solution of value $2k-\sum_{i,j} a_{i,j}$ if and only if $G$ has a maximum cut of total weight $k$. Thus, \MaxQP{} with only negative $A$~entries can be used to solve \MaxCut{} exactly, implying that even this special case is NP-hard. Furthermore, this special case translates to the closely related \textsc{MaxCut Gain} problem~\cite{Charikarwirth04,KhotODonnel09}.

\item \emph{Correlation clustering:}
In correlation clustering~\cite{BansalBC04,CharikarGW05,DemaineEFI06,Swamy04}, we are provided with pairwise judgments of the similarity of $n$ data items. In the simplest version of the problem there are three possible inputs for each pair: similar (\emph{i.e.}\ positive), dissimilar (\emph{i.e.}\ negative), or no judgment. In a given clustering of the $n$ items, a pair of items is said to be in \emph{agreement} (\emph{disagreement}) if it is a positive (negative) pair within one cluster or a negative (positive) pair across two distinct clusters. In \textsc{MaxCorr}, the goal is to maximize the \emph{correlation} of the clustering; that is, the absolute difference between the number of pairs in agreement and the number of pairs in disagreement, across all clusters. Note that when only two clusters are allowed, this directly corresponds to \UnitMaxQP{}, the variant of \MaxQP{} where $a_{i,j} \in \{-1,0,1\}$ for each entry $a_{i,j}$ of~$A$.

\item \emph{Ising spin glass model:}
	Spin glass models are used to in physics to study disordered magnetic phases of matter. Such system are notoriously hard to solve, and various techniques to approximate the free energy were developed. In the Ising spin-glass model~\cite{Barahona1982,Talagrand03}, each node in the graph represents a single spin which can either point up (+1) or down (-1), and neighboring spins $(i, j)$ may have either positive or negative coupling energy $a_{i,j}$ between them. The energy of this system (when there is no external field) is given by its Hamiltonian $H= -1 \cdot \sum_{i,j} a_{i,j}\alpha(i)\alpha(j)$, where $\alpha(i) \in \{-1,1\}$ is the spin at site~$i$. A famous problem in the physics of spin-glasses is the characterization of the ground state --- the state that minimizes the energy of the system. This problem is precisely \MaxQP. 
\end{itemize}

It is convenient to view \MaxQP{} in graph-theoretic terms. Let $G=(V,E)$ be the \emph{graph associated with $A$}, where $V=\{1,\ldots,n\}$ and $E=\{\{i,j\}: a_{i,j} \neq 0 \}$. The first algorithmic result for \MaxQP{} was due to Bieche \emph{et al.}~\cite{Bieche1980} and Barahona \emph{et al.}~\cite{BarahonaMaynard1982} who studied the problem in the context of the Ising spin glass model. They showed that when~$G$ is restricted to be planar, the problem is polynomial-time solvable via a reduction to maximum-weight matching. At the same time, Barahona proved that the problem is NP-hard for three-dimensional grids~\cite{Barahona1982} or apex graphs (graphs with a vertex whose removal leaves the graph planar)~\cite{Barahona1983}. 

\subsection{Approximation Algorithms}

As \MaxQP{} is NP-hard, even for restricted instances, our focus is naturally on polynomial-time approximation algorithms. We note that the fact that the values of $A$ are allowed to be both positive and negative makes \MaxQP{} quite unique in the context of approximation and presents several challenges. First of all, there is an immediate equivalence between \MaxQP{} and the problem of minimizing~\eqref{eqn:maxQP}, as maximizing $\val_x(A)$ is the same as minimizing $\val_x(-1 \cdot A)$. Furthermore, solutions might have negative values; that is, we might have $\val_x(A) < 0$ for certain solutions $x$. This poses an extra challenge since a solution with a non-positive value is not an $f(n)$-approximate solution, for any function $f$, in case the optimum is positive (which it always is whenever $A \neq 0$, see Charikar and Wirth~\cite{Charikarwirth04} and our Lemma~\ref{lem:LowerBound}). In particular, a uniformly at random chosen solution $x$ has $\val_x(A)=0$ on expectation, and unlike \MaxCut{}, such a solution is unlikely to be useful as any kind of approximation.

Alon and Naor~\cite{AlonNaor06} were the first to show that these difficulties can be overcome by carefully rounding a semidefinite relaxation of \MaxQP{}. In particular, they studied the problem when $G$ is bipartite, and showed that using a rounding technique that relies on the famous Grothendieck inequality, one can obtain an approximation factor guarantee of~$\approx 0.56$ for the bipartite case. Later, together with Makarychev and Makarychev~\cite{AlonMMN05}, they showed that the integrality gap of the semidefinite relaxation is $O(\lg \chi(G))$ and $\Omega(\lg \omega(G))$, where~$\chi(G)$ and $\omega(G)$ are the chromatic and clique numbers of~$G$, respectively. In particular, this gap is constant for several interesting graph classes such as $d$-degenerate graphs and $H$-minor free graphs, and it generalizes the previous result of Alon and Naor~\cite{AlonNaor06} as $\chi(G) \leq 2$ when $G$ is bipartite.
\begin{theorem}[\cite{AlonMMN05,AlonNaor06}]
\label{thm:Alon}
\MaxQP{} restricted to graphs of $O(1)$ chromatic number can be approximated within a factor of $\Omega(1)$ in polynomial time. 
\end{theorem}

Regarding the general version of the problem, where $G$ can be an arbitrary graph, an integrality gap of $O(\lg n)$ for the semidefinite relaxation was first shown by Nesterov~\cite{Nesterov98}. However, his proof was non-constructive.
Charikar and Wirth~\cite{Charikarwirth04} made his proof constructive, and provided a rounding procedure for the relaxation that guarantees $\Omega(1/\lg n)$-approximate solutions regardless of the structure of $G$. 
\sloppy
\begin{theorem}[\cite{Charikarwirth04,Nesterov98}]
\label{thm:CharikarWirth}%
\MaxQP{} can be approximated within a factor of $\Omega(1/\lg n)$ in polynomial time. 
\end{theorem}
\fussy

As for the time complexity of the algorithm in Theorems~\ref{thm:Alon} and~\ref{thm:CharikarWirth} above, Arora, Hazan, and Kale~\cite{AroraHK05} provided improved running times for several semidefinite programs, including the relaxation of \MaxQP{}. They showed that this relaxation can be solved (to within any constant factor) in $\tilde{O}(n^{1.5}\cdot \min\{N,n^{1.5}\})$ time, where $N$ is the number of nonzero entries in $A$ and $\tilde{O}$ ignores polylogarithmic factors. Thus, for general matrices $A$, this running time is $O(n^3)$, and for matrices with $O(n)$ nonzero entries this is $O(n^{2.5})$. 

There has also been work on approximation lower bounds for \MaxQP{}. Alon and Naor~\cite{AlonNaor06} showed that \MaxQP{} restricted to bipartite graphs cannot be approximated within $16/17 +\eps$ unless P=NP, while Charikar and Wirth~\cite{Charikarwirth04} showed that, assuming P$\neq$NP, the problem admits no~$(11/13+\eps)$-approximation when $G$ is an arbitrary graph. Both these results follow somewhat directly from the $16/17 +\eps$ lower bound for \MaxCut{}~\cite{Haastad2001}. In contrast, Arora \emph{et al.}~\cite{AroraBKSH05} showed a much stronger lower bound by proving that there exists a constant~$c > 0$ such that \MaxQP{} cannot be approximated within~$\Omega(1/\lg^c n)$, albeit under the weaker assumption that NP $\not \subseteq$ DTime($n^{\lg^{O(1)}n})$.

\subsection{Our results}
\label{sec:our-results}

In this paper we focus on \emph{sparse} graphs, \emph{i.e.}, graphs 
where the number of edges $m$ is $O(n)$. This corresponds to matrices $A$ having $O(n)$ nonzero entries. Note that \MaxQP{} remains APX-hard in this case as well (see Theorem~\ref{thm:hardness-bipartite} in Appendix~\ref{sec:hardness}). Nevertheless, we show that one can abandon the semidefinite approach in favor of simpler, ``purely combinatorial'' algorithms, while still maintaining comparable performances. In particular, our algorithms are faster than than those obtained from the semidefinite approach whose fastest known implementation requires~$O(n^{2.5})$ time~\cite{AroraHK05}. Furthermore, most of them are quite easy to implement.

\subsubsection{Generic classes of sparse graphs}

We begin by considering three basic classes of sparse graphs. 
Our first result concerns bounded degree graphs. 
We show that a simple greedy algorithm for computing a matching in $G$ can be used to obtain an $\Omega(1)$-approximate solution to the corresponding \MaxQP{} instance. 
\begin{theorem}
\label{thm:MaxDegree}
Let $\Delta \geq 1$. \MaxQP{} restricted to graphs of maximum degree~$\Delta$ can be approximated within a factor of~$1/2\Delta$ in $O(n \lg n)$ time.
\end{theorem}

Next we consider $d$-degenerate graphs. Recall that a graph is \emph{$d$-degenerate} if each of its subgraphs has a vertex of degree at most $d$. We show that by using a more elaborate structure than the matching used in the proof of Theorem~\ref{thm:MaxDegree}, one can obtain an $\Omega(1)$ approximation in this setting for the unit weight case. 
\begin{theorem}
\label{thm:Degenerate}
There is a $1/2d$ approximation algorithm for \UnitMaxQP{} restricted to $d$-degenerate graphs with $O(n)$ running time.
\end{theorem}

Finally, we consider \emph{$\delta$-dense} graphs, graphs which have at most $\delta n$ edges for some fixed $\delta$. By a slight modification of the idea used to prove both theorems above, we obtain the following generalization of Theorem~\ref{thm:Degenerate} at the cost of a slight increase in the running time and decrease in the approximation factor guarantee. 
\begin{theorem}
\label{thm:Density}
There is a $1/3\delta$ approximation algorithm for \UnitMaxQP{} restricted to $\delta$-dense graphs without isolated vertices that runs in $O(n^{1.5})$ time.
\end{theorem}

Observe that all the three results above improve on the running time of Theorem~\ref{thm:Alon}. Furthermore, while Theorem~\ref{thm:Alon} provides an $\Omega(1)$ approximation for graphs of bounded degree and bounded degeneracy, this is not true for graphs of bounded density. For example, consider a graph consisting of a clique of size $\sqrt{n}$ together with a perfect matching on the remaining vertices. The result of Alon \emph{et al.}~\cite{AlonMMN05} implies that the semidefinite relaxation has an integrality gap of $O(\lg n)$ on such a graph, while the algorithm in Theorem~\ref{thm:Density} provides an $\Omega(1)$ approximation.

\subsubsection{\texorpdfstring{$H$}{H}-minor free graphs}

We next consider graph classes that exclude certain minors. 
A graph~$G$ is~$H$-minor free for a fixed graph~$H$ if one cannot obtain in~$G$ an isomorphic copy of~$H$ by a series of edge contractions, edge deletions, and vertex deletions. We begin by considering the class of apex-minor free graphs (recall that a graph is \emph{apex} if it contains one specific vertex whose deletion results in a planar graph). This graph class is better known as the class of minor-closed graphs with bounded local treewidth~\cite{Eppstein00}, and includes well-studied classes such as planar and bounded genus graphs. We tailor the approach of Eppstein~\cite{Eppstein00} and Grohe~\cite{Grohe03} for designing approximation algorithms for apex-minor-free graphs to our setting to obtain the following result.
\begin{theorem}
\label{thm:bdltw}
Let~$\eps > 0$. There is an~$O(n)$ time~$(1-\eps)$-approximation algorithm for \MaxQP{} restricted to apex-minor free graphs.
\end{theorem}

We then show that Theorem~\ref{thm:bdltw} can be extended to general $H$-minor free graphs with unit weights, although this requires an extra $O(n)$ factor in the running time. This algorithm is obtained by using a partitioning algorithm of Demaine \emph{et al.}~\cite{DemaineHK05} which is based on the graph minor decomposition of Robertson and Seymour~\cite{RobertsonSeymour03a}. 
\begin{theorem}
\label{thm:h-minor-free}
For~$\eps > 0$ and any graph~$H$ there is an $O(n^2)$ time~$(1-\eps)$-approximation algorithm for \UnitMaxQP{} restricted to~$H$-minor free graphs.
\end{theorem}

\subsubsection{Maximum correlation}

Finally, we note that our results have direct consequences for the \textsc{MaxCorr} problem: Charikar and Wirth~\cite{Charikarwirth04} proved that an $\alpha$-approximation algorithm for \MaxQP{} implies an $\alpha/(2 + \alpha)$-approximation algorithm for \textsc{MaxCorr}. Combining this with the results discussed above gives us the following: 
\begin{corollary}
\textnormal{\textsc{MaxCorr}} can be approximated within a factor of 
\begin{compactitem}
	\item $1/(4d+1)$ on $d$-degenerate graphs in $O(n)$ time;
	\item $1/3-\eps$ on $H$-minor free graphs in~$O(n^2)$ time;
	\item $1/3-\eps$ on apex-minor free graphs in~$O(n)$ time.
\end{compactitem}
\end{corollary}

\section{Preliminaries}
\label{sec:preliminaries}

Throughout the paper we use $G=(V,E)$ to denote the graph associated with our input matrix~$A$; that is, $V=\{1,\ldots,n\}$ and $E=\{\{i,j\} : a_{i,j} \neq 0\}$. Thus, $n=|V|$ and we let $m = |E|$. We slightly abuse notation by allowing a solution $x$ to denote either a vector in $\{-1,1\}^n$ indexed by $V$ or a function $x:V \to \{-1,1\}$. For a solution $x$, we let $\val_x(G)= \sum_{\{i,j\} \in E} a_{i,j}x_ix_j$, and we let $\opt(G)=\max_x \val_x(G)$. We use $||A||$ to denote the sum of absolute values in $A$, \emph{i.e.}, $||A||=\sum_{i,j} |a_{i,j}|$. Note that $\opt(G) \leq ||A||$. 

We use standard graph-theoretic terminology when dealing with the graph $G$, as in \emph{e.g.}~Die\-stel~\cite{Diestel10}. In particular, for a subset $V' \subseteq V$, we let $G[V']$ denote the subgraph of~$G$ \emph{induced} by $V'$; \emph{i.e.}, the subgraph with vertex set $V'$ and edge set $\{\{u,v\} \in E: u,v \in V'\}$. We let~$G-V' = G[V\setminus V']$, and for a subset of edges $E' \subseteq E$ we let $G-E'$ denote the graph $(V,E')$ without isolated vertices. For a pair of disjoint subsets $V_1,V_2 \subseteq V$, we let $E(V_1,V_2)= \{\{u,v\} \in E: u \in V_1, v\in V_2\}$. Finally, we use $N(v)=\{u : \{u,v\} \in E\}$ to denote the \emph{neighborhood} of a vertex $v \in V$.


\subsection{Useful observations}

Note that for a uniformly chosen at random solution $x$, the value $a_{i,j}x_ix_j$ is zero in expectation for any edge $\{i,j\} \in E$. This implies that $\opt(G) \geq 0$. Moreover, a solution $x$ with $\val_x(G) \geq 0$ can be computed in linear time:
\begin{lemma}
\label{lem:ZeroValSol}%
One can compute in $O(m+n)$ time a solution $x$ with $\val_x(G) \geq 0$. 
\end{lemma}
\begin{proof}
For each vertex $i \in V$, let $E(i) = \{\{i,j\} \in E : j < i\}$. Consider an arbitrary initial solution $x$, and let $z_i=\sum_{\{i,j\} \in E(i)} x_ix_ja_{i,j}$. Then $z_i$ is the contribution of edges in~$E(i)$ to $\val_x(G)$. We compute a solution $x^*$ by scanning the vertices from $1$ to $n$. For a given vertex $i$, we check whether~$z_i < 0$. If so, we set $x^*_i = -x_i$, and otherwise we set $x^*_i=x_i$. Note that $z^*_i=\sum_{\{i,j\} \in E(i)}x^*_i x^*_j a_{i,j}$ must now be positive. As the value of $x^*_i$ does not change $z^*_j$ for any $j < i$, when we finish our scan we have $z^*_i \geq 0$ for each $i \in \{1,\ldots,n\}$. Thus, $\val_{x^*}(G) = \sum_i z^*_i \geq 0$. 
\end{proof}

\begin{lemma}
\label{lem:DisjointUnion}%
Let $V_1, V_2 \subseteq V$ be two disjoint subsets of vertices, and let $x^{(1)}$ and~$x^{(2)}$ be two solutions for $G[V_1]$ and $G[V_2]$ of value $z_1$ and $z_2$ respectively. Then at least one of the solutions $x^{(1)} \cup x^{(2)}$ and $-x^{(1)} \cup x^{(2)}$ has value $z_1+z_2$ for $G[V_1 \cup V_2]$.
\end{lemma}

\begin{proof}
Suppose $x^{(1)} \cup x^{(2)}$ has value less than $z_1+z_2$. This means that the total contribution of the edges in $E(V_1,V_2)$ is negative in this solution. Observe that in~$-x^{(1)} \cup x^{(2)}$ each edge of~$E(V_1,V_2)$ with negative contribution under $x^{(1)} \cup x^{(2)}$ now has positive contribution, and vice versa. The lemma thus follows. 
\end{proof}

Combining Lemma~\ref{lem:ZeroValSol} and Lemma~\ref{lem:DisjointUnion} above, we get an important property of $\val(G)$, namely that it is monotone with respect to induced subgraphs. 

\begin{lemma}
\label{lem:monotonic}
Let $H$ be an induced subgraph of $G$. Then given a solution $x$, one can compute in~$O(n+m)$ time a solution $x^*$ for $G$ with $\val_{x^*}(G) \geq \val_x(H)$.
\end{lemma}

\begin{proof}
Let $V_0 \subseteq V$ be the vertices of $G$ which are not present in $H$. According to Lemma~\ref{lem:ZeroValSol} we can compute a solution $x^{(0)}$ for $G[V_0]$ with value at least zero in linear time. According to Lemma~\ref{lem:DisjointUnion} either $x^{(0)} \cup x$ or $x^{(0)} \cup -x$ has value at least $\val_{x^{(0)}}(G[V_0]) + \val_x(H) \geq \val_x(H)$. Thus, taking $x^*$ to be the solution with higher value out of $x^{(0)} \cup x$ or $x^{(0)} \cup -x$ proves the lemma. 
\end{proof}


\section{Matching-Based Algorithms}
\label{sec:MatchingBased}%

In this section we present approximation algorithms for  \MaxQP{} using certain type of matchings that we compute for $G$. In particular, we provide proofs for Theorems~\ref{thm:MaxDegree}, \ref{thm:Degenerate}, and \ref{thm:Density}. The proof of each theorem is given in each subsection below.


\subsection{Graphs of Bounded Degree}
\label{sec:MaxDegree}%

We begin with the case where $G$ has bounded maximum degree~$\Delta$. Our algorithm for this case is quite simple: It greedily computes a matching and then outputs a solution corresponding to this matching. For an edge set~$E' \subseteq E$, we let~$w(E') = \sum_{\{i,j\} \in E'} |a_{i,j}|$ be the total absolute value of the edges in $E'$. To prove that our strategy gives a good approximation on the optimal value of the \MaxQP{} instance, we first observe that, given a matching $M$, we can derive a solution with value at least $w(M)$.
\begin{lemma}
\label{lem:matching-solution}
Given a matching~$M$, one can compute in $O(n)$ time a solution~$x$ with~$\val_x(G) \ge w(M)$.
\end{lemma}
\begin{proof}
Let $V_M=\{u \in V: \{u,v\} \in M\}$ be the set of vertices matched by $M$ in $G$. By Lemma~\ref{lem:monotonic} it suffices to compute a solution~$x$ for the set of endpoints~$V_M$ in~$M$ such that~$\val_x(G[V_M]) \ge w(M)$. We construct~$x$ by induction on~$t = |M|$.	For~$t=1$, let~$\{u, v\} \in M$. If~$a_{u, v} > 0$, we choose~$x_u$ and~$x_v$ to be of equal value, and otherwise we choose~$x_u$ and~$x_v$ to be of opposite value.	Then~$\val_x(G[\{u, v\}]) = |a_{i,j}| = w(M)$.	Suppose now that~$t > 1$ and let~$M_0 = M \setminus \{e_t\}$. By induction, we have a solution~$x^{(0)}$ for $V_{M_0}$ with~$\val_{x^{(0)}}(G[V_{M_0}]) \ge w(M_0)$. Let~$e_t = \{u, v\}$ and let~$x^{(t)} \colon \{u, v\} \to \{-1, 1\}$ such that~$\val_{x^{(t)}}(G[\{u, v\}]) = w(e_t)$ as in the case of~$t=1$.
Then, by Lemma~\ref{lem:DisjointUnion}, either $x^{(0)} \cup x^{(t)}$ or~$x^{(0)} \cup -x^{(t)}$ has value at least~$w(M_0) + w(e_t) = w(M)$, and we are done.
\end{proof}

For our algorithm we choose a particular matching~$M^*$ which we obtain greedily as follows: We first sort the edges $\{u,v\} \in E$ in non-increasing order according to their absolute values $|a_{u,v}|$. We then iteratively choose the first edge among the remaining edges, and remove all edges that share an endpoint with this edge, and stop when no more edges are left. The matching $M^*$ is the set of all edges selected in this process.

\begin{lemma}
\label{lem:matching-bound}
$w(M^*) \ge 1/2\Delta \cdot w(E)$.
\end{lemma}
\begin{proof}
Let~$M^* = \{e_1, \dots, e_t\}$, with~$w(e_1) \ge w(e_2) \ge \dots \ge w(e_t)$, and let~$E_i$ be the set of edges that were removed after choosing~$e_i$. Formally, $E_i = \{e \in E \colon e \cap e_i \ne \emptyset\} \setminus \bigcup_{j=1}^{i-1} E_j$.	Let~$e_i = \{u_i, v_i\}$. As all edges in~$E_i$ are incident to either~$u_i$ or~$v_i$, we have that~$|E_i|$ is at most the sum of the degrees of~$u_i$ and~$v_i$, which is at most~$2\Delta$.	As~$e_i$ is the edge with the highest weight within~$E_i$ (otherwise we would have chosen a different edge for~$e_i$), we have
\[
w(\{e_i\}) \ge \frac{w(E_i)}{|E_i|} \ge \frac{w(E_i)}{2\Delta}. 
\]
Since~$E_1, \dots, E_t$ is a partition of all edges in~$G$, we get
\[
w(M) = \sum_{i=1}^k w(\{e_i\}) \ge \sum_{i=1}^k \frac{w(E_i)}{2\Delta} = \frac{w(E)}{2\Delta}.\qedhere
\]
\end{proof}

\begin{proof}[Proof of Theorem~\ref{thm:MaxDegree}]
The matching $M^*$ can be computed in $O(n \lg n)$ time. Due to Lemma~\ref{lem:matching-solution} and Lemma~\ref{lem:matching-bound}, we obtain from $M^*$ a solution $x$ with 
\[
\val_x(G) \ge 1/2\Delta \cdot w(E) = 1/2\Delta \cdot ||A||
\] 
in $O(n)$ time.	Since~$\opt(G) \le ||A||$, this solution is~$1/2\Delta$-approximate.
\end{proof}

\subsection{Graphs of Bounded Degeneracy}
\label{sec:Degen}

We next present our approximation algorithm for \UnitMaxQP{} restricted to $d$-degenerate graphs. Our algorithm extends the algorithm of the previous subsection by considering a structure more elaborate than a matching.

Let us begin with introducing some terminology. Let $x$ be a solution for $G$. If for some edge~$e= \{u,v\}$ we have $a_{u,v}x_ux_v > 0$, then we say that $e$ is \emph{good} in $x$, otherwise we say that it is \emph{bad} in $x$. Let~$\{u, v, w\} \subseteq V$ be the vertices of a triangle in~$G$. Note that there exists a solution~$x$ for $G$ where all all edges in this triangle are good iff~$a_{u,v} \cdot a_{v, w} \cdot a_{w, u} = 1$. Hence, we call a triangle \emph{good} if~$a_{u,v} \cdot a_{v, w} \cdot a_{w, u} = 1$, and \emph{bad} otherwise. 

We next introduce a special type of subgraph that will be exploited by our algorithm. Recall that a graph is a \emph{split graph} if its vertex set can be partitioned into a clique and an independent set. We call an $n$-vertex split graph \emph{basic} if it is connected, and the class of vertices inducing a clique is of size 2 (note that this clique need not be maximal). A subgraph $G'$ of $G$ is \emph{easy} if it is a basic split graph without any bad triangles. In this way, an easy subgraph of $G$ contains one \emph{center edge} whose endpoints are adjacent to all remaining \emph{outside vertices}, and the outside vertices have no edges between them. Note that, possibly, an outside vertex can form a triangle with the center edge. In this case, the triangle must be good. An \emph{easy packing} of~$G$ is a family~$\mathcal F = \{F_1, \dots F_t\}$ of pairwise disjoint subsets of vertices such that each~$F_i$ induces an easy subgraph in~$G$. For each $F_i \in \mathcal{F}$, we use $m(F_i)$ to denote the number of edges in $G[F_i]$. Furthermore, we let~$m(\mathcal F)= \sum_i m(F_i)$ and~$V_\mathcal F = \bigcup_i F_i$.

\begin{lemma}
\label{lem:easy-pack-val}
Given an easy packing~$\mathcal F$, one can compute in $O(n)$ time a solution~$x$ with~$\val_x(G) \ge m(\mathcal F)$.
\end{lemma}
\begin{proof}
By Lemma~\ref{lem:monotonic}, it suffices to compute a solution~$x$ for~$V_\mathcal F$ of value at least~$m(\mathcal F)$. Let~$\mathcal F = \{F_1, \dots, F_t\}$. We construct such a solution by induction on~$t$.	For~$t = 1$, let~$\{u, v\}$ be the center edge of the easy subgraph~$G[F_1]$.	We assign to~$u$ an arbitrary value. To each neighbor~$w$ of~$u$ we assign the same value if~$a_{u, v} = 1$ and the opposite value otherwise. Afterwards, we proceed analogously with the neighbors of~$v$. This ensures that each edge incident to~$u$ or~$v$ contributes~$+1$ to~$\val_x(G[F_1])$.	As each triangle in~$G[F_1]$ is good and for each such triangle, two edges contribute~$+1$ to the value, the third edge must contribute~$+1$ as well. Thus, $\val_x(G[F_1]) = m(F_1) = m(\mathcal F)$. Suppose then that~$t > 1$ and let~$\mathcal F_0 = \mathcal F \setminus \{F_t\}$. By induction, we have a solution~$x^{(0)}$ for~$V_{\mathcal F_0}$ with~$\val_{x^{(0)}}(G[V_{\mathcal F_0}]) \ge m(\mathcal F_0)$. Let~$x^{(t)} \colon F_t \to \{1, -1\}$ be such that~$\val_{x^{(t)}}(G[F_t]) \ge |E(G[F_t])|$ as in the case of~$t = 1$.	Then, by Lemma~\ref{lem:DisjointUnion}, either~$x^{(0)} \cup x^{(t)}$ or~$x^{(0)} \cup -x^{(t)}$ have value at least~$m(\mathcal F_0) + |E(G[F_t])| = m(\mathcal F)$, and we are done.
\end{proof}

For our algorithm we need a particular easy packing which is  inclusion-wise maximal. Moreover, we want to be able to compute this packing efficiently so as not to exceed the linear running time promised in Theorem~\ref{thm:Degenerate}. We next describe a procedure that achieves just that.

\begin{quote}
\vspace{.5em}
\noindent \textsc{Algorithm EasyPack}:
\begin{compactenum}
\item Compute an inclusion-wise maximal matching~$M$ in $G$, and let~$I$ be the set of unmatched vertices. 
\item Let~$M^* = \emptyset$ and $I^*=I$.
\item For each~$\{x, y\} \in M$, if there are two vertices~$u, v \in I^*$ such that both~$u, x, y$ and~$v, x, y$ induce a triangle in~$G$, then 
\begin{compactitem}
    \item $M^*= M^* \cup \{\{u,x\},\{v,y\}\}$.
    \item $I^* = I^* \setminus \{u,v\}$.
\end{compactitem}
Otherwise, $M^*= M^* \cup \{\{x,y\}\}$.
\item Initialize the easy packing with the matching $M^*$, that is, $\mathcal F^* = \{\{x, y\} : \{x, y\} \in M^*\}$.
\item For each~$v \in I^*$, if there is an edge~$\{x, y\} \in M^*$ such that~$v, x, y$ induce a path or a good triangle, then remove $v$ from $I^*$ and add it to the vertex set~$F \in \mathcal F^*$ that contains~$x$ and~$y$.
\end{compactenum}
\end{quote}

\begin{lemma}
\label{lem:max-easy-pack}%
The algorithm above computes in $O(n)$ time an easy packing~$\mathcal F^*$ with~$m(\mathcal F^*) \ge |V_{\mathcal F^*}|/2$.
\end{lemma}
\begin{proof}
Clearly, $F^*$ is an easy packing. Steps 1, 2, and 4 can be performed in linear time in a straightforward manner. Step 3 can be executed in linear time by first adding a marker to every vertex in~$I$ and then, computing~$N(x) \cap N(y)$ for each~$\{x, y\} \in M$ in $O(\Delta^2)=O(1)$ time. If the intersection contains two marked vertices, we remove their marks, and add the corresponding edges to~$M^*$. Finally, step 5 can be performed in linear time by storing for each matched vertex its partner and then iterating over the neighborhood of each~$v \in I^*$. Thus, the entire algorithm can be executed in linear time. To complete the proof, observe that each $G[F_i]$, $F_i \in \mathcal{F}$, is connected and consists of at least one edge. This implies that $m(\mathcal F^*) \geq |V_{\mathcal F^*}|/2$, and so the lemma holds.
\end{proof}

Next we need to provide an upper bound on~$\opt(G)$.
For this we first make an observation on the matching~$M^*$ that we computed in Step 3 of our algorithm above.

\begin{lemma}
\label{lem:matching-triangles}
Each edge in~$M^*$ forms at most one triangle with a vertex from~$I^*$.
\end{lemma}
\begin{proof}
Let~$\{x, y\} \in M^*$. If this edge is also in~$M$, then it forms a triangle with at most one vertex from~$I$ due to Step 3 of the algorithm above. Otherwise, one of its endpoints, say~$x$, was originally unmatched, that is~$x \in I$. Assume that there is a vertex~$v \in I^* \subseteq I$ such that~$v, x, y$ form a triangle.
Then~$\{v, x\} \notin M$ but both~$v$ and~$x$ were in~$I$, a contradiction to the maximality of~$M$ (see Step 1).
\end{proof}

Now we are ready to bound~$\opt(G)$ from above.
Recall that~$V_{\mathcal F^*}$ is the set of vertices in the packing computed by our algorithm above.

\begin{lemma}
\label{lem:opt-vs-unpacked}
$\opt(G) \le d \cdot |V_{\mathcal F^*}|$, where~$d$ is the degeneracy of~$G$.
\end{lemma}
\begin{proof}
Let~$M^*$ be the matching and~$\mathcal F^*$ be the easy packing computed above. Then~$I^* = V \setminus V_{\mathcal F^*}$. Let~$T$ be the set of triangles that contain a vertex from~$I^*$. We first show that the triangles in~$T$ are edge-disjoint. Observe that~$I^*$ is an independent set, as~$I^* \subseteq I$ is a subset of the vertices that were left unmatched in step 1; an edge within this vertex set would contradict the maximality of the matching~$M$ computed in this step. So, if two triangles in~$T$ share an edge, then it must be an edge in~$M^*$. But this is impossible according to Lemma~\ref{lem:matching-triangles}, and so the triangles in $T$ are indeed edge-disjoint.

Next observe that if a vertex $v$ is in $I^*$ at the end of the algorithm above, we know that either~$v$ is adjacent neither to~$x$ nor to~$y$, or~$\{v, x, y\}$ induces a bad triangle in $G$. Thus, $T$ consists solely of bad triangles. Furthermore, for each~$v \in I^*$, there are~$\deg(v)/2$ bad triangles that contain~$v$, where $\deg(v)$ is the degree of $v$ in $G$. Thus, overall there are at least~$|T| = \sum_{v \in I^*} \deg(v)/2$ bad triangles in~$G$. As the triangles in~$T$ do not share edges, we obtain that there are at least~$|T|$ bad edges in every solution for~$G$. Consequently,
\[
\opt(G) \le (||A||-|T|)-|T| = m-2|T| = m - \sum_{v \in V'} \deg(v).
\]
Now, as $G[V_{\mathcal F^*}]$ is a $d$-degenerate graph, it has at most $d \cdot |V_{\mathcal F^*}|$ edges. Thus, 
\begin{align*}
2\opt(G) &\le 2m - 2\sum_{v \in I^*} \deg(v) = \sum_{v \in V} \deg(v) - 2\sum_{v \in I^*} \deg(v) \\
&\le \sum_{v \in V_{\mathcal F^*}} \deg(v) \le 2d \cdot |V_{\mathcal F^*}|,
\end{align*}
and the lemma follows.
\end{proof}

\begin{proof}[Proof of Theorem~\ref{thm:Degenerate}]
By Lemma~\ref{lem:easy-pack-val} and Lemma~\ref{lem:max-easy-pack} we can compute in linear time an easy packing~$\mathcal F^*$ and a solution~$x$ with~$\val_x(G) \ge |V_{\mathcal F^*}|/2$.
Using Lemma~\ref{lem:opt-vs-unpacked} we obtain
\[
\val_x(G) \ge |V_{\mathcal F^*}|/2 \ge \opt(G)/2d,
\]
and the theorem is proven.
\end{proof}

\subsection{Graphs of Bounded Density}
\label{sec:Density}

We next turn to handle graphs of bounded density $\delta$ without isolated vertices. Here we use a particular kind of easy packings described in section \ref{sec:Degen}, were each vertex subset induces a star in $G$. A \emph{star packing} of $G$ is a family of pairwise disjoint subsets of vertices $\mathcal{F}=\{F_1,\ldots,F_t\}$ such that each $G[F_i]$ is a star. Note that since a star is a basic split graph without bad triangles, a star packing is also an easy packing. We compute a specific star $\mathcal{F}^*$ packing inspired by the work of Erd\H{o}s, Gy\'{a}rf\'{a}s, and Kohayakawa~\cite{ErdosGK97} concerning graph cuts.  

To compute $\mathcal{F}^*$, we first compute a maximum matching $M$ in $G$. Note that we require $M$ to be of maximum cardinality, and not only maximal inclusion-wise. Let $V_M \subseteq V$ denote the vertices in $G$ that are matched in $M$, and let $I=V \setminus V_M$. We initialize the packing $\mathcal{F}^*$ with the matching $M$, that is, $\mathcal F^* = \{\{x, y\} : \{x, y\} \in M\}$. We then iterate through vertices $v \in I$, and add $v$ to any subset $F$ such that $G[F \cup \{v\}]$ is a star. Let $V_{\mathcal{F}^*}=\bigcup_{F \in \mathcal{F}^*}F$ denote the set of vertices in the packing at the end of this process, and let $I^*= V \setminus V_{\mathcal{F}^*}$.  

\begin{lemma}
\label{lem:StarPack}%
Let $F \in \mathcal{F}^*$. Then there is at most one vertex in $v \in I^*$ which is adjacent to any vertex in $F$.    
\end{lemma}
\begin{proof}
Let $\{x,y\} \subseteq F \cap V_M$ be the edge in $M$ from which $F$ was initialized. Then, as~$F \setminus \{x, y\} \subseteq I$, any vertex in $I^*$ can only be adjacent to $x$ and $y$ in $F$. Let $I^*(x) \subseteq I^*$ and $I^*(y) \subseteq I^*$ respectively denote the set of neighbors of $x$ and $y$ in $I^*$. If both $I^*(x) = \emptyset$ and $I^*(y) = \emptyset$ we are done. So assume w.l.o.g. that $I^*(x) \neq \emptyset$ and let $v \in I^*(x)$. Now if there exists a vertex $u \in I^*(y) \setminus I^*(x)$, then replacing $\{x,y\}$ with $\{x,v\}$ and $\{y,u\}$ in $M$ gives us a matching $M'$ with $|M'| > |M|$, contradicting the fact that $M$ is a maximum matching. Thus, $I^*(y) \setminus I^*(x) = \emptyset$. Furthermore, if $v \notin I^*(y)$ then $G[F \cup \{v\}]$ is a star, and so $v$ would have been added to $F$ by the algorithm above. Thus, $I^*(x) = I^*(y)$. Finally, if there are two distinct vertices $u,v \in I^*(x)$, then $u,v \in I^*(y)$, and again we can replace $\{x,y\}$ with $\{x,v\}$ and $\{y,u\}$ in $M$. This implies that $|I^*(x) \cup I^*(y)| \leq 1$, and so the lemma is proven. 
\end{proof}

\begin{lemma}
\label{lem:LowerBound}%
$m(\mathcal{F}^*) \geq m/3\delta$. 
\end{lemma}

\begin{proof}
Since each $F \in \mathcal{F}^*$ induces a star in $G$ we have $2m(\mathcal{F}^*) \geq |V_{\mathcal{F}^*}|$. Moreover, as by Lemma~\ref{lem:StarPack}, each $F \in \mathcal{F}^*$ is adjacent to at most one vertex in $I^*$, and each vertex in $I^*$ is adjacent to at least one $F \in \mathcal{F}^*$ as $G$ has no isolated vertices, we have $|\mathcal{F}^*| \geq |I^*|$. Since each $F \in \mathcal{F}^*$ induces at least one edge, we have $m(\mathcal{F}^*) \geq |\mathcal{F}^*|$, and so $m(\mathcal{F}^*) \geq |I^*|$. Thus, combining all of the above, we get
\[
3m(\mathcal{F}^*) \geq |V_{\mathcal{F}^*}|+|I^*|=n \geq m /\delta,
\]
and the lemma is proven.
\end{proof}

\begin{proof}[Proof of Theorem~\ref{thm:Density}]
Due to Lemma~\ref{lem:easy-pack-val} and~\ref{lem:LowerBound}, the packing $\mathcal{F}^*$ yields a solution $x$ with $\val_x(G) \geq m/3\delta$. Since $\opt(G) \leq ||A||=m$, this solution is $1/3\delta$-approximate. The running time for computing~$\mathcal F^*$ is dominated by the computation of the maximum matching $M$ for the initial packing, taking~$O(m\sqrt n) = O(n^{1.5})$ time~\cite{MicaliVazirani}; computing~$\mathcal F^*$ from $M$ can be done in $O(m+n)=O(n)$ time. 
\end{proof}


\section{Algorithms Based On Treewidth Partitions}

In this section we present approximation algorithms for sparse \MaxQP{} instances that exclude certain types of minors. Namely, we prove Theorems~\ref{thm:bdltw} and~\ref{thm:h-minor-free}. Our algorithms all revolve around the Baker technique for planar graphs~\cite{Baker1994} and its generalizations~\cite{DemaineHK05,Eppstein99,Grohe03}, all using what we refer to here as a \emph{treewidth partition} --- a partition of the vertices of $G$ into $V_0,\ldots,V_{k-1}$ such that $G-V_i$ has bounded treewidth for any subset $V_i$ in the partition. As treewidth plays a central role here, we begin with formally defining this notion.

A tree decomposition is a pair~$(\mathcal{T}, \mathcal{X})$ where~$\mathcal X$ is a family of vertex subsets of~$G$, called \emph{bags}, and~$\mathcal T$ is a tree with~$\mathcal X$ as its node set. The decomposition is required to satisfy (i) $\{X \in \mathcal X : v \in X\}$ is connected in~$\mathcal T$ for each~$v \in V$, and (ii) for each~$\{u, v\} \in E$ there is a bag~$X \in \mathcal X$ that contains both~$u$ and $v$. The \emph{width} of a tree decomposition~$(\mathcal T, \mathcal X)$ is~$\max_{X \in \mathcal X} |X|-1$, and the \emph{treewidth} of~$G$ is the smallest width amongst all its tree decompositions. The proof of the following lemma, relying on standard dynamic programming techniques, is deferred to Appendix~\ref{sec:treewidth}.
\begin{lemma}
\label{lem:boundedTW}%
\MaxQP{} restricted to graphs of treewidth at most~$k$ can be solved in~$2^{O(k)} \cdot n$ time.
\end{lemma}
\subsection{Apex-minor free instances}
\label{sec:bdltw}

Our starting point for the~$(1-\eps)$-approxima\-tion for \MaxQP{} on~$H$-minor free graphs with~$H$ being an apex graph (Theorem~\ref{thm:bdltw}) is a layer decomposition $L_0,\ldots,L_\ell \subseteq V$ of $G$. Here $L_0=\{v\}$ for some arbitrary vertex $v \in V$, $L_i =\{ u: d(v,u)=i\}$ are all vertices at distance $i$ from $v$, for each $i \in \{1,\ldots,\ell\}$. This is the standard starting point of all Baker-type algorithms, and can be computed in linear time via breadth-first search from~$v$. Note that $L_0,\ldots,L_\ell$ form a partition of $V$, and that for each $i \in \{0,\ldots,\ell\}$, each vertex in $L_i$ has neighbors only in $L_{i-1} \cup L_i \cup L_{i+1}$ (here and elsewhere in this section we set $L_{-1}=L_{\ell+1}=\emptyset$ when necessary). 

Given $0 < \eps \leq 1$, we let $k$ be the smallest integer such that~$4/k \leq \eps$. For each $i \in \{0,\ldots,k-1\}$, let~$\mathcal{L}_i$ denote the union of all vertices in layers with index equal to $i (\mmod k)$; that is, $\mathcal{L}_i = \bigcup_{j \equiv i (\mmod k)} L_j$. We define two subgraphs of $G$: The graph~$G_i$ is the graph induced by $V - \mathcal{L}_i$, and the graph $H_i$ is the graph induced by $N[\mathcal{L}_i]$. Note that there is some overlap between the vertices of $G_i$ and~$H_i$, but each edge of~$G$ appears in exactly one of these subgraphs. Also note that since there is an apex graph~$H$ that~$G$ does not contain as a minor, $G$ and each of its minors have bounded local treewidth~\cite{Eppstein00}; thus both~$G_i$ and~$H_i$ are bounded treewidth graphs~\cite{Grohe03}.

Our algorithm computes $k$ different solutions for $G$, and selects the best one (\emph{i.e.}, the solution~$x$ with highest~$\val_x(G)$) as its solution. For $i \in \{0,\ldots,k-1\}$, we first compute an optimal solution for $G_i$ in linear time using the algorithm given in Lemma~\ref{lem:boundedTW}. We then extend this solution to a solution $x^{(i)}$ for $G$ as is done in Lemma~\ref{lem:monotonic}. In this way we obtain in linear time $k$ solutions $x^{(0)},\ldots,x^{(k-1)}$ with $\val_{x^{(t)}}(G) \geq \opt(G_i)$ for each $i \in \{0,\ldots,k-1\}$. In Lemma~\ref{lem:localTwLayers} we argue that the solution of maximum objective value is $(1-\eps)$-approximate to the optimum of $G$; the proof of Theorem~\ref{thm:bdltw} will then follow as a direct corollary.

\begin{lemma}
\label{lem:localTwLayers}
There is a solution $x \in \{x^{(0)},\ldots,x^{(k-1)}\}$ with $\val_x(G) \geq (1-\eps) \cdot \opt(G)$.
\end{lemma}

\begin{proof}
Let $x^*$ denote the optimal solution for $G$. Then, as the edge set of $G$ is partitioned into the edges of $G_i$ and $H_i$, we have 
\[
\opt(G_i)+\opt(H_i) \geq \val_{x^*}(G_i)+\val_{x^*}(H_i)  = \val_{x^*}(G) =  \opt(G)
\]
for each $i \in \{0,\ldots,k-1\}$. Next observe that any two subgraphs $H_{i_1}$ and $H_{i_2}$ with $|i_1-i_2| \geq 4$ do not have vertices in common, nor are there any edges between these two subgraphs in $G$. It follows that for any $j \in \{0,1,2,3\}$, the graph $\bigcup_{i \equiv j (\mmod 4)} H_i$  is an induced subgraph in $G$, and so $\opt(G) \geq \sum_{i \equiv j (\mmod 4)} \opt(H_i)$ by Lemma~\ref{lem:monotonic}. Thus, we have 
\[
4 \cdot \opt(G) \geq \sum_j \sum_{i \equiv j (\mmod 4)} \opt(H_i) = \sum_i \opt(H_i). 
\]
Combining the two inequalities above we get  
\begin{align*}
\sum_{i=0}^{k-1} \val_{x^{(i)}}(G) \geq \sum_{i=0}^{k-1} \opt(G_i) &\ge k \cdot \opt(G) - \sum_{i=0}^{k-1} \opt(H_i)\\
&\ge k \cdot \opt(G) - 4 \cdot \opt(G) = (k-4) \opt(G).
\end{align*}
It follows that the best solution out of $x^{(0)},\ldots,x^{(k-1)}$ has value at least $(1-4/k)\cdot\val(G)$, which is at least  $(1-\eps)\cdot\val(G)$ since~$4/k \le \eps$. 
\end{proof}

\subsection{General \texorpdfstring{$H$}{H}-minor free instances}

To obtain the~$(1-\eps)$-approximation for \UnitMaxQP{} on~$H$-minor free graphs for any fixed graph~$H$ (Theorem~\ref{thm:h-minor-free}), we make use of the lower bound obtained in Section~\ref{sec:Density}, our algorithm for \MaxQP{} restricted to bounded-treewidth graphs, and the following theorem by Demaine \emph{et al.}~\cite[Theorem 3.1]{DemaineHK05}:
\begin{theorem}[\cite{DemaineHK05}]
\label{thm:MinorFreePartition}
For a fixed graph $H$, there is a constant $c_H$
such that, for any integer $k \geq 1$ and for every $H$-minor free
graph $G$, the vertices of $G$ can be
partitioned into $k$ sets such that the graph obtained by taking the union of any $k-1$ of these sets has treewidth at most $c_H \cdot k$. Furthermore, such a partition can be found in polynomial time.
\end{theorem}

Note that this theorem gives a partition similar to the one used in the previous subsection, albeit slightly weaker. In particular, there is no restriction on the edges connecting vertices in different subsets of the partition as was the case in the previous subsection. It is for this reason that arbitrary weights are difficult to handle, and we need to resort to the lower bound of Lemma~\ref{lem:LowerBound}. Fortunately, for the unweighted case, we can use the fact that there exists some constant $h$ depending only on $H$ such that $G$ has at most $hn$ edges (see \emph{e.g.}~\cite{Diestel10}). In particular, it can be shown that $h = O(n' \sqrt{\lg n'})$~\cite{Kostochka84}, where $n'$ is the number of vertices of $H$. Combining this fact with Lemma~\ref{lem:LowerBound}, we get:
\begin{lemma}
\label{lem:MinorFree}
$\opt(G) \geq m/3h$.
\end{lemma}


Our algorithm proceeds as follows. Fix $k \geq 6h/\eps$, and let $V_0,\ldots,V_{k-1}$ denote the partition of $V$ computed by the algorithm from Theorem~\ref{thm:MinorFreePartition}. For each  $i \in\{0,\ldots,k-1\}$, let $E_i$ denote the set of edges $E(V_i,\bigcup_{j \neq i} V_j)$, and let $m_i=|E_i|$. Furthermore, let $G_i=G-E_i$. As both $G[V_i]$ and $G[V \setminus V_i]$ have bounded treewidth, we can compute an optimal solution for each of these subgraphs (and therefore also for $G_i$) using the algorithm in Lemma~\ref{lem:boundedTW}. Using Lemma~\ref{lem:DisjointUnion}, we can extend the optimal solutions for $G[V_i]$ and $G[V \setminus V_i]$ to a solution~$x^{(i)}$ for $G$ with value 
\[
\val_{x^{(i)}}(G) \geq \opt(G_i).
\]
On the other hand, the optimal solution of $G$ cannot do better than 
\[
\opt(G_i) + m_i \geq \opt(G).
\]

Combining the two inequalities above, we can bound the sum of the objective values obtained by all our solutions by 
\begin{align*}
\sum_{i=0}^{k-1} \val_{x^{(i)}}(G) &\geq \sum_{i=0}^{k-1} \opt(G_i) \geq \sum_{i=0}^{k-1} (\opt(G)-m_i)\\
&= \sum_{i=0}^{k-1} \opt(G)-2m \geq (k-6h) \cdot \opt(G),
\end{align*}
where the last inequality follows from Lemma~\ref{lem:MinorFree}. Thus at least one of these solutions has value at least $(k-6h)/k \cdot \opt(G)$, which is greater than $(1-\eps)\opt(G)$ by our selection of parameter~$k$. 

To analyze the time complexity of our algorithm, observe that computing each solution~$x^{(i)}$ requires $O(n)$ time according to Lemma~\ref{lem:boundedTW} and Lemma~\ref{lem:DisjointUnion}. Thus, the time complexity of the algorithm is dominated by the time required to compute the partition promised by Theorem~\ref{thm:MinorFreePartition}. Demaine \emph{et al.}~\cite{DemaineHK05} showed that this partition can be computed in linear time given the graph decomposition promised by Robertson and Seymour's graph minor theory~\cite{RobertsonSeymour03a}. In turn, Grohe~\emph{et al.}~\cite{GroheKR13} presented an $O(n^2)$-time algorithm for this decomposition, improving earlier constructions~\cite{DemaineHK05,KawarabayashiW11}. Thus, the total running time of our algorithm can also be bounded by $O(n^2)$. This completes the proof of Theorem~\ref{thm:h-minor-free}. 

\section{Conclusion}
We presented efficient combinatorial approximation algorithms for sparse instances of \MaxQP{} without resorting to the semidefinite relaxation, as done by Alon and Naor~\cite{AlonNaor06} and Charikar and Wirth~\cite{Charikarwirth04}. 
From a theoretical perspective, we still leave open whether there is a fast algorithm for~$d$-degenerate \MaxQP{} instances which obtains an $\Omega(1)$ approximation factor guarantee. Even more interesting is to design a purely combinatorial algorithm for general \MaxQP{} instances with an approximation guarantee of $\Omega(1/\lg n)$.
Finally, the simplicity of our algorithms compels the study of their usability in practice, especially for characterizations of ground states of spin glass models.

\bibliographystyle{plain}
\bibliography{bib}

\appendix

\section{An Exact Algorithm for Bounded Treewidth Graphs}
\label{sec:treewidth}

We prove Lemma~\ref{lem:boundedTW} by presenting an algorithm for \MaxQP{} restricted to graphs of treewidth at most~$k$ running in~$2^{O(k)} \cdot n$ time. For this we require the concept of nice tree decompositions~\cite{Kloks94}.

A tree decomposition~$(\mathcal T, \mathcal X)$ is \emph{rooted} if there is a designated bag~$R \in \mathcal X$ being the root of~$\mathcal T$. A rooted tree decomposition is \emph{nice} if each bag~$X \in \mathcal X$ is either (i) a leaf node ($X$ contains exactly one vertex and has no children in~$\mathcal T$), (ii) an introduce node ($X$~has one child~$Y$ in~$\mathcal T$ with~$Y \subset X$ and~$|X \setminus Y| = 1$), (iii) a forget node ($X$~has one child in~$Y$ in~$\mathcal T$ with~$X \subset Y$ and~$|Y \setminus X| = 1$), or (iv) a join node ($X$ has two children~$Y, Z$ in~$\mathcal T$ with~$X = Y = Z$). Given a tree decomposition, one can compute a corresponding nice tree decomposition with the same width in linear time~\cite{Kloks94}.

Our algorithm employs the standard dynamic programming technique on nice tree decompositions.

\begin{proof}[Proof of Lemma~\ref{lem:boundedTW}]
Let~$(\mathcal T, \mathcal X)$ be a nice tree decomposition of~$G$ of width~$k$ with root bag~$R$. For a node~$X \in \mathcal X$ let~$\mathcal T_X$ be the subtree of~$\mathcal T$ rooted at~$X$.
Furthermore, let~$G_X$ be the subgraph of~$G$ induced by the vertices in the bags of~$\mathcal T_X$ (while~$G[X]$ is the subgraph of~$G$ induced only by the vertices in~$X$).
We describe a table in which we have an entry~$D[X, x]$ for each bag~$X \in \mathcal X$ and for each solution~$x: X \to \{-1,1\}$. The entry~$D[X, x]$ contains the value of an optimum solution for~$G_X$, where the values of the vertices in~$X$ are fixed by the solution~$x$.

If~$X$ is a \emph{leaf node}, then~$G_X$ contains no edges and so~$D[X, x] = 0$. If~$X$ is an \emph{introduce node}, then let~$v \in X \setminus Y$ be the introduced vertex, where~$Y$ is the child of~$X$ in $\mathcal T$, and let~$x \setminus x_v$ be the solution~$x$ restricted to the vertices of~$Y$.
Then~$D[X, x]$ additionally contains the value of all edges incident to~$v$, that is,
$$
D[X, x] = D[Y, x \setminus x_v] + \sum_{u \in N(v)} x_u x_v a_{u, v}.
$$
If~$X$ is a \emph{forget node}, then let~$v \in Y \setminus X$ be the forgotten vertex, where~$Y$ is the child of~$X$ in~$\mathcal T$. Then, every value except for~$x_v$ is set in~$x$, so we must choose it so that the value is maximized. 
Then
$$
D[X, x] = \max_{x_v: v \to \{-1, 1\}} D[Y, x \cup x_v].
$$
Finally, if~$X$ is a \emph{join node}, then let~$Y$ and~$Z$ be the children of~$X$ in~$\mathcal T$. Note that~$D[Y, x] + D[Z, x]$ contains the value of~$G[X]$ twice, so
$$
D[X, x] = D[Y, x] + D[Z, x] - \val_x(G[X]).
$$

The tree decomposition contains~$O(n)$ nodes, and for each node there are at most~$O(2^k)$ solutions; thus we need to compute $O(2^k \cdot n)$ entries~$D[\cdot, \cdot]$, each of which can be computed in $O(k)$ time. The optimum value is the maximum over all~$O(2^k)$ solutions for the root bag~$R$. As we can compute a tree decomposition of width~$O(k)$ in~$2^{O(k)} \cdot n$ time~\cite{BodlaenderDDFLP16}, we can compute~$\opt(A)$ in~$O(2^{O(k)} \cdot n + 2^k\cdot k \cdot n) = 2^{O(k)} \cdot n$ time. 
\end{proof}

\section{A Hardness Result}
\label{sec:hardness}
Alon and Naor~\cite{AlonNaor06} show that \MaxQP{} restricted to bipartite graphs is not approximable in polynomial time with a ratio of~$16/17 + \eps$ unless~P$=$NP. Using the same idea, we show that this approximation lower bound also holds for \UnitMaxQP{} on~$2$-degenerate bipartite graphs.

\begin{theorem}
\label{thm:hardness-bipartite}%
If \UnitMaxQP{} restricted to $2$-degenerate bipartite graphs can be approximated within a factor of~$(16/17+\eps)$ in polynomial time, then P$=$NP.
\end{theorem}

\begin{proof}
We reduce from unweighted \MaxCut{} which does not admit a~$(16/17+\eps)$-approximation unless P$=$NP~\cite{Haastad2001}. Given an undirected unweighted graph~$G = (V, E)$, we construct a graph~$G' = (V \cup V', E')$ by subdividing each edge in~$E$, that is, for every edge~$\{u, w\} \in E$ we add a vertex~$v$ to~$V'$ and the edges~$\{u, v\}, \{v, w\}$ to~$E'$. One edge has weight~$1$ while the other edge has weight~$-1$.	Clearly, $G'$ is bipartite; $V'$ is one bipartition.	As all vertices in~$V'$ have degree two, $G'$ is $2$-degenerate as well.

Let~$x$ be a solution for~$G'$.	Observe that for every vertex~$v \in V'$ we can assume that at least one of its incident edges contributes positively to~$\val_x(G')$; otherwise multiply $x_v$ by $-1$. Furthermore, note that the cut in~$G$ corresponding to~$x$ (restricted to~$V$) is of size~$\val_x(G')/2$: If both edges incident to~$v$ contribute positively to~$\val_x(G')$, then the edge in~$G$ subdivided by~$v$ is cut. Otherwise, the two edges contribute~$0$ to~$\val_x(G')$, and the corresponding edge in~$G$ is not cut.

It follows that if there is a~$(16/17+\eps)$-approximation for \MaxQP{}, then there is one for \MaxCut{} as well, implying P$=$NP by~\cite{Haastad2001}. 
\end{proof}
\end{document}